\def\icalp{0}

\ifnum\icalp=1
\documentclass[a4paper,UKenglish,cleveref, autoref]{lipics-v2019}

\else
\documentclass[a4paper,UKenglish,11pt]{article}
\fi

\ifnum\icalp=0
\usepackage{fullpage}
\usepackage[colorlinks,urlcolor=blue,citecolor=blue,linkcolor=blue]{hyperref}
\usepackage{algorithm}
\usepackage{algorithmic}
\usepackage{soul}
\fi
\usepackage{amsmath}
\usepackage{amssymb}
\usepackage{amsthm}
\usepackage{xcolor}
\usepackage{xfrac}

\def\Drawings{1}
\ifnum\Drawings=1
\usepackage{caption}
\usepackage{pgf,tikz,pgfplots}
\pgfplotsset{compat=1.14}
\usepackage{mathrsfs}
\usetikzlibrary{arrows}
\usepackage{tikz}
\usepackage{tkz-euclide}
\usetikzlibrary{positioning}
\usepackage{mathtools}
\usepackage{MnSymbol}
\usetikzlibrary{calc}
\usetkzobj{all}
\fi

\newtheorem{thm}{Theorem}
\newtheorem{tcor}{Corollary}
\newtheorem{lem}{Lemma}[section]
\newtheorem{cor}[lem]{Corollary}
\newtheorem{prop}[lem]{Proposition}
\ifnum\icalp=0
\newtheorem{claim}[lem]{Claim}
\fi

\newtheorem{rem}[lem]{Remark}

\newtheorem{dfn}[lem]{Definition}

\newtheorem{ntn}[lem]{Notation}
\newtheorem{obs}[lem]{Observation}

\newcommand{\dft}[1]{\textbf{\textit{#1}}}

\newcommand{\e}{\varepsilon}
\newcommand{\eps}{\varepsilon}

\newcommand{\calG}{\mathcal{G}}
\newcommand{\calL}{\mathcal{L}}
\newcommand{\N}{\mathbf{N}}

\newcommand{\E}{\mathbf{E}}

\newcommand{\cAER}{\mA_{ua}}
\newcommand{\cAERR}{\mA_{ba}}

\DeclareMathOperator{\disj}{disj}

\newcommand{\confEqn}[1]{
\ifnum\icalp=0
	\begin{equation}
	#1
	\end{equation}
\else
	$#1$
\fi
}

\newcommand{\abs}[1]{\left|#1\right|}

\newcommand{\paren}[1]{\left(#1\right)}
\newcommand{\set}[1]{\left\{#1\right\}}
\newcommand{\sqb}[1]{\left[#1\right]}
\newcommand{\sucht}{\,\middle|\,}
\renewcommand{\th}{^{\textrm{th}}}

\newcommand{\thres}{\theta}
\newcommand{\setthres}{4\alpha\lceil\log n\rceil/\eps}
\newcommand{\setbeta}{\eps/2\lceil\log n\rceil}
\newcommand{\setell}{\lceil \log n \rceil}
\newcommand{\setL}{\frac{\log n}{C \cdot \loglog n}}
\newcommand{\eqdef}{\stackrel{\text{def}}{=}}

\newcommand{\thr}{\theta}

\ifnum\icalp=0

\newcommand{\sal}{\hyperref[sal]{\textup{\color{black}{\sf Sample-a-leaf}}}}

\newcommand{\se}{\hyperref[se]{\textup{\color{black}{\sf Sample-edge}}}}
\newcommand{\rw}{\hyperref[rw]{\textup{\color{black}{\sf Random-walk}}}}
\else

\newcommand{\sal}{\text{\sf Sample-a-leaf}}

\newcommand{\se}{\text{\sf Sample-edge}}
\newcommand{\rw}{\text{\sf Random-walk}}
\fi

\DeclareMathOperator{\poly}{poly}

\renewcommand{\deg}{{\sf deg}}
\DeclareMathOperator{\nbr}{{\sf nbr}}
\DeclareMathOperator{\pair}{{\sf pair}}
\DeclareMathOperator{\answer}{{\sf ans}}

\newcommand{\mA}{\mathcal{A}}

\newcommand{\HP}{P}

\newcommand{\pset}{\mathcal{X}}

\newcommand{\mT}{\mathcal{T}}

\newcommand{\FAIL}{{\sc FAIL}}
\newcommand{\HEADS}{{\sc HEADS}}

\newcommand{\calE}{\mathcal{E}}
\DeclareMathOperator{\cost}{cost}

\newcommand{\cA}{\mathcal{A}}
\newcommand{\calQ}{\mathcal{Q}}
\newcommand{\calT}{\mathcal{T}}
\newcommand{\calP}{\mathcal{P}}

\newcommand{\davg}{{d}}
\newcommand{\ua}{\sqrt m/\davg}

\newcommand{\ba}{\alpha/\davg}
\newcommand{\baf}{\frac{\alpha}{\davg}}

\newcommand{\tT}{T}
\newcommand{\tV}{V}
\newcommand{\tE}{E}

\newcommand{\lT}{\widehat{T}}

\newcommand{\depth}{\Delta}

\newcommand{\lb}{L}

\newcommand{\aT}{{T'}}
\newcommand{\aTt}{{T'_t}}

\newcommand{\aV}{{V'}}
\newcommand{\aE}{{E'}}

\newcommand{\Eone}{{{\cal E}_1}}
\newcommand{\Etwo}{{{\cal E}_2}}
\newcommand{\Ethree}{{{\cal E}_3}}

\newcommand{\PhiOne}{\Phi^{T' \rightarrow T}_{r' \rightarrow u_*}}
\newcommand{\PhiTwo}{\Phi^{T' \rightarrow T}_{u'_* \rightarrow r}}
\newcommand{\PhiThree}{\Phi^{T\rightarrow T''}_{r \rightarrow u'_*}}

\newcommand{\loglog}{{\rm loglog}}

\newcommand{\dmax}{{d_{\max}}}

\definecolor{bubbles}{rgb}{0.91, 1.0, 1.0}
\definecolor{aqua}{rgb}{0.0, 1.0, 1.0}
\definecolor{blue-green}{rgb}{0.0, 0.87, 0.87}
\definecolor{amethyst}{rgb}{0.6, 0.4, 0.8}

\DeclareRobustCommand{\hldana}[1]{{\sethlcolor{aqua}\hl{#1}}}
\DeclareRobustCommand{\hltalya}[1]{{\sethlcolor{yellow}\hl{#1}}}
\DeclareRobustCommand{\hlwill}[1]{{\sethlcolor{pink}\hl{#1}}}
\newcommand{\dnote}[1]{{\hldana{[\textbf{D}: #1]}}}

\ifnum\icalp=0
\newcommand{\tnote}[1]{\hltalya{[\textbf{T}: #1]}}
\else
\renewcommand{\tnote}[1]{\hltalya{[\textbf{T}: #1]}}
\fi

\newcommand{\wnote}[1]{\hlwill{[\textbf{W}: #1]}}

\title{The Arboricity Captures the Complexity of Sampling Edges}

\ifnum\icalp=0
\author{%
  Talya Eden\thanks{School of Electrical Engineering, Tel Aviv University, Tel Aviv, Israel}
  \and
  Dana Ron\thanks{School of Electrical Engineering, Tel Aviv University, Tel Aviv, Israel}
  \and
  Will Rosenbaum\thanks{Max Planck Institute for Informatics, Saarbr\"{u}cken, Germany}
}
\date{\today}
\else

\author{Talya Eden}{Tel Aviv University, Tel Aviv, Israel}{talyaa01@gmail.com}{Supported by the Azrieli fellowship program for graduate students, by the Sephora Scholarship and by the Weinstein Graduate Studies Prize.}{}{}
\author{Dana Ron}{Tel Aviv University, Tel Aviv, Israel}{danaron@tau.ac.il}{Partially supported by the Israel Science Foundation grant No.~1146/18}{}{}
\author{Will Rosenbaum}{Max Planck Institute for Informatics, {Saarbr\"{u}cken}, Germany}{will.rosenbaum@gmail.com}{}{}
\authorrunning{T. Eden, D. Ron and W. Rosenbaum}

\Copyright{Talya Eden, Dana Ron and Will Rosenbaum}



\category{}

\relatedversion{}

\supplement{}

\EventEditors{John Q. Open and Joan R. Access}
\EventNoEds{2}
\EventLongTitle{42nd Conference on Very Important Topics (CVIT 2016)}
\EventShortTitle{CVIT 2016}
\EventAcronym{CVIT}
\EventYear{2016}
\EventDate{December 24--27, 2016}
\EventLocation{Little Whinging, United Kingdom}
\EventLogo{}
\SeriesVolume{42}
\ArticleNo{23}

\fi

\newcommand{\DrawTreeLB}{

\pgfdeclarelayer{bg}    
\pgfsetlayers{bg,main}  
	
\begin{tikzpicture}[trim left,trim right=0pt,xshift=-3.3cm, very thick,yscale=0.6,xscale=.6, every node/.style={scale=0.6} ,font=\huge]

\tkzDefPoint[label=above left:$T$](0,0){A}
\tkzDefPoint[label=left:$L$](-2.5,-5){B}
\tkzDefPoint(2.5,-5){C}
\tkzDefPoint[label=left:$D$](-4,-8){D}
\tkzDefPoint(4,-8){E}
\tkzDrawPolygon[](A,B,D,E,C)
\tkzDrawPolygon(A,B,C,A)

\node[circle,draw, minimum size=1cm] (O) at  (0,-1.2) {$r$};
\node[circle,draw, minimum size=.5cm] (u_star) at  (-.75,-4) {$u_*$};

\begin{scope}[xshift=12cm]
\tkzDefPoint[label=above left:$T'$](0,0){A'}
\tkzDefPoint[label=left:$L$](-2.5,-5){B'}
\tkzDefPoint(2.5,-5){C'}
\tkzDefPoint(-4,-8){D'}
\tkzDefPoint(4,-8){E'}
\tkzDrawPolygon(A',B',C',A')

\node[circle,draw, minimum size=.5cm] (O') at  (0,-1.2) {$r'$};

\node[circle,draw, minimum size=.5cm] (up_star) at  (-.75,-4) {$u'$};

\draw (up_star) edge [-latex,bend right] node[above, xshift=-4cm, yshift=.3cm]{} (O) ;

\draw (O') edge [-latex,bend right]  node[above, xshift=-3.25cm, yshift=-2.2cm]{$\Phi$} (u_star) ;
\draw (O') [decorate, decoration=zigzag, densely dotted] -> (up_star.north);

\draw (O) [decorate, decoration=zigzag, densely dotted] -> (u_star.north) node[midway,above] (p) {};

\draw (O') [decorate, decoration=zigzag, densely dotted] -> (up_star.north) node[midway, above] (p') {} ;

\end{scope}

\draw[draw=none, use as bounding box](-2,0) rectangle (20,2);

\begin{pgfonlayer}{bg}    

\draw (p') edge [-latex,bend right, dashed, gray] node[above, xshift=-4cm, yshift=.3cm]{} (p) ;
    \end{pgfonlayer}

\end{tikzpicture}

}

\begin{document}

\maketitle

\begin{abstract}
  In this paper, we revisit the problem of sampling edges in an unknown graph $G = (V, E)$ from a distribution that is (pointwise) almost uniform over $E$. We consider the case where there is some a priori upper bound on the arboriciy of $G$.  Given query access to a graph $G$
  over $n$ vertices
  and of average degree $\davg$ and arboricity at most $\alpha$, we design an algorithm that performs
  $O\!\left(\baf\cdot \frac{\log^3 n}{\eps}\right)$ queries in expectation and returns an edge in the graph such that \emph{every} edge $e \in E$ is sampled with probability $(1 \pm \e)/m$.  The algorithm performs two types of queries: degree queries and neighbor queries.  We show that the upper bound is tight (up to poly-logarithmic factors and the dependence in $\eps$), as $\Omega\!\left(\baf\right)$ queries are necessary for the easier task of sampling edges from any distribution over $E$ that is close to uniform in total variational distance. We also prove that even if $G$ is a tree (i.e., $\alpha = 1$ so that $\baf=\Theta(1)$), $\Omega\!\left(\frac{\log n}{\loglog n}\right)$ queries are necessary to sample an edge from any distribution that is pointwise close to uniform, thus establishing that a $\poly(\log n)$ factor is necessary for constant $\alpha$.
  Finally we show how our algorithm can be applied to obtain a new result on approximately counting subgraphs, based on the recent work of Assadi, Kapralov, and Khanna (ITCS, 2019).
\end{abstract}

\ifnum\icalp=0
\thispagestyle{empty}
\newpage
\thispagestyle{empty}
\tableofcontents
\newpage
\setcounter{page}{1}
\fi

\section{Introduction}
\label{sec:introduction}

Let $G=(V,E)$ be a graph over $n$ vertices and $m$ edges.
We consider the problem of sampling an edge in $G$ from a pointwise almost uniform distribution over $E$. That is, for each edge $e\in E$, the probability  that $e$ is returned is $(1\pm \eps)/m$, where $\eps$ is a given approximation parameter.
An algorithm for performing this task has random access to the vertex set $V = \{1,\dots,n\}$
and can perform queries to $G$.
The allowed queries are (1) \emph{degree queries} denoted $\deg(v)$ (what is the degree, $d(v)$, of a given vertex $v$) and (2) \emph{neighbor queries} denoted $\nbr(v,i)$ (what is the $i\th$ neighbor of $v$).\footnote{If $i>d(v)$ then a special symbol, e.g. $\bot$, is returned. } We refer to this model as \emph{the uniform vertex sampling model}.

Sampling edges almost uniformly is a very basic sampling task. In particular it gives the power to sample vertices with probability approximately proportional to their degree, which is a  useful primitive.
Furthermore, there are sublinear algorithms that are known to work when given access to uniform edges
(e.g.,~\cite{Assadi2018}) and can be adapted to the case when the distribution over the edges is almost uniform (see Section~\ref{subsec:intro-application} for details).
An important observation is that in many cases it is crucial that the sampling distribution is pointwise-close to uniform rather than close with respect to the Total Variation Distance (henceforth TVD) --
see the discussion in~\cite[Sec. 1.1]{Eden2018-sosa}.

Eden and Rosenbaum~\cite{Eden2018-sosa} recently showed that $\Theta^*(\ua)$ queries are both sufficient and necessary for sampling edges almost uniformly,
where $\davg = 2m/n$ denotes the average degree in the graph.
(We use the notation $O^*$ to suppress factors that are polylogarithmic in $n$ and polynomial in $1/\e$.)
The instances for which the task is difficult (i.e., for which $\Omega(\ua)$ queries are necessary), are characterized by having very dense subgraphs, i.e., a subgraph with average degree $\Theta(\sqrt m)$. Hence, a natural question is whether it is possible to achieve lower query complexity when some a apriori bound on the density of subgraphs is known. A well studied measure for bounded density (``everywhere'') is the graph \emph{arboricity} (see Definition~\ref{def:arboricity} below).
Indeed there are many natural families of graphs that have bounded arboricity such as graphs of bounded degree, bounded treewidth or bounded genus, planar graphs, graphs that exclude a fixed minor and many other graphs.
In the context of social networks, preferential attachment graphs  and additional generative models exhibit bounded arboricity~\cite{barabasi1999emergence, baur2007generating,bauer2010enumerating}, and this has also been empirically validated for many real-world graphs~\cite{goel2006bounded, eppstein2011listing, shin2018patterns}.

We describe a new algorithm for sampling edges almost uniformly whose runtime is $O^*(\ba)$ where $\alpha$ is an upper bound on the arboricity of $G$. In the extremal case that $\alpha = \Theta(\sqrt{m})$, the runtime of our algorithm is the same as that of~\cite{Eden2018-sosa} (up to poly-log factors). For smaller $\alpha$, our algorithm is strictly faster. In particular for $\alpha = O(1)$, the new algorithm is \emph{exponentially} faster than that of~\cite{Eden2018-sosa}. We also prove matching lower bounds, showing that for all ranges of $\alpha$, our algorithm is query-optimal, up to polylogarithmic factors and the dependence in $1/\eps$.

Furthermore, while not as simple as the algorithm of~\cite{Eden2018-sosa}, our algorithm is still easy to implement and does not incur any large constants in the query complexity and running time, thus making in it suitable for practical applications.

\subsection{Problem definition}

In order to state our results precisely, we define the notion of pointwise-closeness of probability distributions (cf.~\cite{Eden2018-sosa}) and arboricity of a graph.

\begin{dfn}
  \label{dfn:pointwise-close}
  Let $D$ be a fixed probability distribution on a finite set $X$. We say that a probability distribution $\widehat{D}$ is \dft{pointwise  $\e$-close to $D$} if for all $x \in X$,
  \[
  \abs{\widehat{D}(x) - D(x)} \leq \e D(x)\,,\quad\text{or equivalently}\quad 1 - \e \leq \frac{\widehat{D}(x)}{D(x)} \leq 1 + \e\,.
  \]
  If $D = U$, the uniform distribution on $X$, then we say that $\widehat{D}$ is \dft{pointwise $\e$-close to uniform}.
\end{dfn}
For the sake of conciseness, from this point on, unless explicitly stated otherwise, when we refer to an edge sampling algorithm, we mean an algorithm that returns edges according to a distribution that is pointwise-close to uniform.

\begin{dfn}\label{def:arboricity}
  Let $G = (V, E)$ be an undirected graph. A forest $F = (V_F, E_F)$ (i.e., a graph containing no cycles) with vertex set $V_F = V$ and edge set $E_F \subseteq E$ is a \dft{spanning forest} of $G$. We say that a family of spanning forests $F_1, F_2, \ldots, F_k$ \dft{covers} $G$ if $E = \bigcup_{i = 1}^k E_{F_i}$. The \dft{arboricity} of $G$, denoted $\alpha(G)$, is the minimum $k$ such that there exists a family of spanning forests of size $k$ that covers $G$.
\end{dfn}

An edge-sampling algorithm is given as input an approximation parameter  $\eps \in (0,1)$ and a parameter $\alpha$ which is an upper bound on the arboricity of $G$.
The algorithm is required to sample 
edges according a distribution that
is pointwise $\eps$-close to uniform. To this end the algorithm is given query access to $G$.
In particular we consider the aforementioned uniform vertex sampling model.

\subsection{Results}
We prove almost matching upper and lower bounds on the query complexity of sampling an edge
according to a distribution that is pointwise-close to uniform when an upper bound on the arboricity of the graph is known. The first lower bound stated below (Theorem~\ref{thm:general-lb}) holds even for the easier task of sampling from a distribution that is close to uniform in TVD.

\begin{thm}
  \label{thm:upper-bound}
  There exists an algorithm $\cA$
  that for any $n$, $m$, $\alpha$, and graph $G = (V, E)$ with $n$ nodes, $m$ edges, and arboricity at most $\alpha$,
  satisfies the following.
  Given $n$ and $\alpha$, $\cA$ returns an edge $e \in E$ sampled from a distribution $\widehat{U}$ that is pointwise $\e$-close to uniform using $O\left(\baf\cdot \frac{\log^3 n}{\eps}\right)$ degree and  neighbor queries in expectation.
\end{thm}

In  Section~\ref{subsubsec:intro-alg} we provide a high-level presentation of
the algorithm referred to in Theorem~\ref{thm:upper-bound}
 and shortly discuss
 how it differs from  the algorithm in~\cite{Eden2018-sosa} (for the case that there is no given upper bound on the arboricity).

We next state our lower bound, which matches the upper bound in Theorem~\ref{thm:upper-bound} up to
a polylogarithmic dependence on $n$  (for constant $\eps$).

\begin{thm}
  \label{thm:general-lb}
    Fix $\e \leq 1/6$ and let $n,m$ and $\alpha$ be parameters such that $\alpha = \alpha(n) \leq \sqrt m$ and $m\leq n\alpha$. Let $\calG^\alpha_{n,m}$ be the family of graphs with $n$ vertices, $m$ edges and arboricity at most $\alpha$. Then any algorithm $\cA$ that for any $G \in \calG^\alpha_{n,m}$ samples edges in $G$ from a distribution that is $\e$-close to uniform in total variation distance---and in particular, any distribution that is pointwise $\e$-close to uniform---requires $\Omega\paren{\ba}$ queries in expectation.
\end{thm}


When $\alpha$ is a constant, then (assuming that $m = \Omega(n)$)
the lower bound in Theorem~\ref{thm:general-lb} is simply $\Omega(1)$, while
Theorem~\ref{thm:upper-bound} gives an upper bound of $O(\log^3 n)$ (for constant $\epsilon$).
We prove that an almost linear dependence on $\log n$ is necessary, even for the case that $\alpha = 1$
(where the graph is a tree).

\begin{thm}
  \label{thm:tree-lb}
  Fix $\e \leq 1/6$, and let $\calT_n$ be the family of trees on $n$ vertices.
  Then any algorithm $\cA$ that for any $G \in \calT_n$ samples edges in $G$ from a distribution that is pointwise $\e$-close to uniform requires $\Omega\paren{\frac{\log n}{\loglog n}}$ queries in expectation.
\end{thm}
We note that both of our lower bounds also hold when the algorithm is also given access to \emph{pair queries} $\pair(u,v)$ (is there an edge between $u$ and $v$).

\subsection{Discussion of the results}\label{subsec:intro-res}

\ifnum\icalp=0
\newcommand{\ourpar}[1]{\paragraph{#1}}
\else
\newcommand{\ourpar}[1]{\smallskip\noindent{\bf #1}~}
\fi

\ourpar{The arboricity captures the complexity of sampling edges.}
The simplest algorithm for sampling edges uniformly is based on rejection sampling.
Namely, it repeats the following until an edge is output: Sample a uniform vertex $u$, flip a coin with bias $d(u) / \dmax$, where $\dmax$ the maximum degree, and if the outcome is {\HEADS}, then output a random edge $(u,v)$ incident to $u$.
 The expected complexity of rejection sampling grows like $\dmax/\davg$, that is, linearly with the maximum degree. As noted earlier, Eden and Rosenbaum~\cite{Eden2017} show that this dependence on the maximum degree is not necessary (for approximate sampling), as $O^*(\sqrt{m}/\davg)$ queries and time always suffice, even when the maximum degree is not bounded (e.g., is $\Theta(n)$). However, if the maximum degree is bounded, and in particular if $\dmax = o(\sqrt{m})$, then rejection sampling has lower complexity than the~\cite{Eden2017} algorithm
 (e.g., in the case that the graph is close to regular, so that $\dmax = O(\davg)$, we get complexity $O(1)$).

 Since the arboricity of a graph is both upper bounded by $\dmax$ and by $\sqrt{m}$, our algorithm can be viewed as ``enjoying both worlds''.
 Furthermore, our results can be viewed as showing that the appropriate complexity measure for sampling edges is not the maximum degree but rather the maximum average degree (recall that the arboricity $\alpha$ measures the maximum density of any subgraph of $G$ -- for a precise statement, see Theorem~\ref{thm:nash-williams}).

\ourpar{Approximately counting the number of edges in bounded-arboricity graphs.}
As shown by Eden, Ron and Seshadhri~\cite{Eden2017}, $\Theta^*(\ba)$ is
also the complexity of estimating the number of edges in a graph when given a bound $\alpha$ on the arboricity of the graph. 
Their algorithm improves on the previous known bound of $O^*\left(\ua\right)$ by Feige~\cite{Feige2006} and Goldreich and Ron~\cite{Goldreich2008} (when the arboricity is $o(\sqrt{m})$). However,
other than the complexity, our algorithm for sampling edges and the algorithm of~\cite{Eden2017}
for estimating the number of edges do not share any similarities,
in particular, as the result of~\cite{Eden2017} is allowed to ``ignore'' an $\eps$-fraction of the graph edges.

Furthermore, while
the complexity of sampling and of approximate counting of edges are the same
(up to $\log n$ and $1/\eps$ dependencies), 
we have preliminary results showing that for other subgraphs this is not necessarily the case. Specifically, there exist graphs with constant arboricity for which estimating the number of triangles can be done using $O^*(1)$ queries, but pointwise-close to uniform sampling requires $\Omega(n^{1/4})$ queries in expectation.

\ourpar{On the necessity of being provided with an upper bound on the arboricity.}
While our algorithm does not require to be given any bound on the average degree $\davg$, it
must be provided with an upper bound $\alpha$ on the arboricity of the given graph.
To see why this is true, consider
the following two graphs. The first graph consists of a perfect matching between its vertices,
so that both its average degree and its arboricity are $1$.
For $\tilde{\alpha}>1$, the second graph consists of a perfect matching over $n - n/\tilde{\alpha}$ vertices and a $\tilde{\alpha}$-regular
graph over the remaining $n/\tilde{\alpha}$ vertices. This graph has an average degree of roughly $2$, and arboricity $\tilde{\alpha}$. If an edge-sampling algorithm is not provided with an appropriate upper bound on the arboricity, but is still required to run in (expected) time that grows like the ratio between the arboricity and the average degree, then it means it can be used to distinguish between the two graphs. However, assuming a random labeling of the vertices of the two graphs, this cannot be done in time $o(\tilde{\alpha})$.

\ourpar{Pointwise closeness vs. closeness with respect to the TVD.}
The lower bound of Theorem~\ref{thm:general-lb} holds for sampling from a distribution that is close to uniform with respect to TVD, and \emph{a fortiori} to sampling from pointwise almost uniform distributions. In contrast, the lower bound of Theorem~\ref{thm:tree-lb} does not apply to sampling edges from a distribution that is $\e$-close to uniform in TVD. Indeed, a simple rejection sampling procedure (essentially ignoring all nodes with degrees greater than $1/\e$) can sample edges from a distribution that is $\e$-close to uniform in TVD using $O(1/\e)$ queries in expectation. Thus, Theorem~\ref{thm:tree-lb} gives a separation between the tasks of sampling from distributions that are pointwise-close to uniform versus close to uniform in TVD. The general upper and lower bounds of Theorems~\ref{thm:upper-bound} and~\ref{thm:general-lb} show that the separation between the complexity of these tasks can be at most poly-logarithmic for any graph.

\subsection{An application to approximately counting subgraphs}\label{subsec:intro-application}
In a recent paper~\cite{Assadi2018}, Assadi, Kapralov, and Khanna made significant progress on the question of counting arbitrary subgraphs in a graph in sublinear time.
Specifically, they provide an algorithm that estimates the number of occurrences of any arbitrary subgraph $H$ in $G$, denoted by $\#H$, to within a $(1\pm \eps )$-approximation with high probability. The running time of their algorithm is  $O^*\left(\frac{m^{\rho(H)}}{\#H}\right)$, where $\rho(H)$ is the fractional edge cover of $H$.\footnote{The fractional edge cover of a graph $H=(V_H, E_H)$ is a mapping $\psi:E_H \rightarrow [0,1]$ such that for each vertex $a \in V_H$, $\sum_{e \in E_H, a \in e} \psi(e) \geq 1.$ The fractional edge-cover number $\rho(H)$ of $H$ is the minimum value
of $\sum_{e \in E_H} \psi(e)$ among all fractional edge covers $\psi$.}
Their algorithm assumes access to uniform edge samples in addition to degree, neighbor and pair queries.
As noted in~\cite{Assadi2018}, their algorithm can be adapted to work with edge samples that are pointwise $\eps$-close to uniform (where this is not true for edge samples that are  only $\eps$-close to uniform in TVD---e.g., when all the occurrences of $H$ are induced by an $\eps$-fraction of the edges). Invoking the algorithm of~\cite{Assadi2018}, and replacing each edge sample with an invocation of  \se\ results in the following corollary.
\begin{tcor}
	Let $G$ be a graph $G = (V, E)$ with $n$ nodes, $m$ edges, and arboricity at most $\alpha$.
	There exists an algorithm that, given $n, \alpha, \eps\in (0,1)$, a subgraph $H$ and query access to $G$, returns a $(1\pm\eps)$ approximation of the number of occurrences of $H$ in $G$, denoted $\# H$.  The expected query complexity and running time of the algorithm are
	\[ O^*\left(\min\left\{m,\frac{n \alpha \cdot m^{\rho(H)-1}}{\# H}\right\}\right)
\;\;\;\mbox{and}\;\;\;
   O^*\left(\frac{n \alpha \cdot m^{\rho(H)-1}}{\# H}\right),
	\]
	respectively,
where $\rho(H)$ denotes the fractional edge cover of $H$, and the allowed queries are degree, neighbor and pair queries.
\end{tcor}
Thus, by combining our result with~\cite{Assadi2018}, we extend the known results for approximately counting the number of subgraphs in a graph in the uniform vertex sampling model.
Furthermore, for graphs in which $m =\Theta(n\alpha)$, we obtain the same query complexity and running time of~\cite{Assadi2018} without the assumption that the algorithm has access to uniform edge samples.

\subsection{A high-level presentation of the algorithm and lower bounds}
While our results concern \emph{un}directed graphs $G= (V,E)$, it will be helpful to
view each edge $\set{u,v} \in E$ as a pair of \emph{ordered} edges $(u,v)$ and $(v,u)$.

\subsubsection{The algorithm}\label{subsubsec:intro-alg}

Sampling an (ordered) edge (almost) uniformly is equivalent to sampling a vertex with probability (almost) proportional to its degree. Hence we focus on the latter task. A single iterations of the  algorithm we describe either returns a vertex or outputs \FAIL. We show that the probability that it outputs \FAIL\ is not too large, and that conditioned on the algorithm returning a vertex, each vertex $v$ is returned with probability proportional to its degree up to a factor of $(1\pm \eps)$.

Our starting point is a structural decomposition result for graphs with bounded arboricity (Lemma~\ref{lem:layered-partition}).
Our decomposition defines a partition of the graph's vertices into \emph{levels} $L_0, L_1, \ldots, L_\ell$. For parameters $\theta$ and $\beta$, $L_0$ consists of all vertices with degree at most $\theta$, and for $i>0$, level $L_i$ contains all vertices $v$ that do not belong to previous levels $L_0,\dots, L_{i-1}$, but have at least $(1-\beta) d(v)$ neighbors in these levels. We prove that that for any graph with arboricity at most $\alpha$, for $\theta = \Theta(\alpha\log n/\eps)$ and $\beta = \Theta(\eps/\log n)$, there exists such a partition into layers with $\ell = O(\log n)$ levels.
We stress that the algorithm does not actually construct such a partition, but rather we use the partition in our analysis of the algorithm.\footnote{This decomposition is related to the forest decomposition of Barenboim and Elkin~\cite{BarenboimElkin2010}.
The main difference, which is essential for our analysis, is that the partition we define is based on the number of neighbors that a vertex has in lower levels \emph{relative\/} to its degree, while in~\cite{BarenboimElkin2010} the partition is based on the absolute number of neighbors to higher levels.
}

In order to gain intuition about the algorithm and its analysis, suppose that all vertices in $L_0$ have degree exactly $\theta$, 
and that all edges in the graph are between vertices in consecutive layers. Consider the following {\em random walk\/} algorithm. The algorithm first selects an index $j\in [0,\ell]$ uniformly at random. It then selects a vertex $u_0$ uniformly at random. If $u_0 \in L_0$, then it performs a random walk of length $j$ starting from $u_0$ (otherwise it outputs \FAIL). If the walk did not pass through any vertex in $L_0$ (with the exception of the starting vertex $u_0$), then the algorithm returns the final vertex reached.

First observe that for every $u \in L_0$, the probability that $u$ is returned is $\frac{1}{\ell+1} \cdot \frac{1}{n}$ (the probability that the algorithm selected $j=0$ and selected $u$ as $u_0$). This equals $\frac{d(u)}{(\ell+1)\cdot \theta\cdot n}$ (by our assumption that $d(u) = \theta$ for every $u \in L_0$).  Now consider a vertex $v \in L_1$. The probability that $v$ is returned is at least $\frac{1}{\ell+1}\cdot \frac{(1-\beta)d(v)}{n}\cdot \frac{1}{\theta} = \frac{(1-\beta)d(v)}{(\ell+1)\cdot \theta\cdot n}$ (the probability that the algorithm selected $j=1$, then selected one of $v$'s neighbors $u \in L_0$, and finally selected to take the edge between $u$ and $v$). In general, our analysis shows that for every $i$ and every $v \in L_i$, the probability that $v$ is returned is at least $\frac{(1-\beta)^i d(v)}{(\ell+1)\cdot \theta\cdot n}$. On the other hand, we show that for every vertex $v$, the probability that $v$ is returned is at most $\frac{d(v)}{(\ell+1)\cdot \theta\cdot n}$. By the choice of $\theta$ and $\beta$ we get that each vertex $v$ is returned with probability in the range $[(1-\eps)d(v)\rho(\eps,n),d(v)\rho(\eps,n)]$ for $\rho(\eps,n) = \Theta(\eps/(\alpha n\log^2 n))$.  By repeating the aforementioned random-walk process until a vertex is returned--- $\Theta\!\left(\frac{\alpha n}{m}\cdot \frac{\log^2n}{\eps}\right) = \Theta\!\left(\baf\cdot \frac{\log^2n}{\eps}\right)$ times in expectation---we obtain a vertex that is sampled with probability proportional to its degree, up to $(1\pm \eps)$.

We 
circumvent the assumption that $d(u) = \theta$ for every $u \in L_0$ by rejection sampling: In the first step, if the algorithm samples $u_0 \in L_0$, then it continues with probability $d(u)/\theta$ and fails otherwise. The assumption that all edges are between consecutive levels is not necessary for the analysis described above to hold. The crucial element in the analysis is that for every vertex $v \notin L_i$, where $i>0$, at least $(1-\beta)$ of the neighbors of $v$ belong to $L_0,\dots,L_{i-1}$. This allows us to apply the inductive argument for the lower bound on the probability that $v$ is returned when we average over all choices of $j$ (the number of steps in the random walk).
For  precise details of the algorithm and its analysis, see Section~\ref{sec:upper-bound}.

\ourpar{On the relation to~\cite{Eden2018-sosa}.}
We briefly discuss the relation between our algorithm for bounded-arboricity graphs, which we denote by $\cAERR$ and the algorithm presented in~\cite{Eden2018-sosa} (for the case that no upper bound is given on the arboricity),
which we denote by $\cAER$.
The algorithm $\cAER$ can be viewed as considering a partition of the graph vertices  into just two
layers according to a degree threshold of roughly $\sqrt{m}$.  
It performs a random walk similarly to $\cAERR$, but where the walk has either length $0$ or $1$.
This difference in the number of layers and the length of the walk, is not only quantitative.
Rather, it allows $\cAER$  to determine to which layer does a vertex belong simply according to its degree.
 This is not possible in the case of $\cAERR$ (with the exception of vertices in $L_0$).
Nonetheless, despite the apparent ``blindness'' of $\cAERR$
  to the layers it traverses in the random walk, we can show the following:
 Choosing the length of the random walk uniformly at random and halting in case the walk returns to $L_0$, ensures that each vertex is output with probability approximately proportional to its degree.


\subsubsection{The lower bounds}

\ourpar{The lower bound of \boldmath{$\Omega\!\left(\baf\right)$} for general \boldmath{$\alpha$}.}
In order to prove Theorem~\ref{thm:general-lb}, we employ the method of~\cite{LB-CC:ER18} (which builds upon the paradigm introduced in~\cite{Blais2012}) based on communication complexity. The idea of the proof is to reduce from the two-party communication complexity problem of computing the disjointness function. The reduction is such that (1) any algorithm that samples edges from an almost-uniform distribution reveals the value of the disjointness function with sufficiently large probability, and (2) every allowable query can be simulated in the two-party communication setting using little communication.

\ourpar{The lower bound of \boldmath{$\Omega\!\left(\frac{\log n}{\loglog n}\right)$} for \boldmath{$\alpha = 1$}.}
As opposed to the proof of the lower bound for general $\alpha$, in the case of $\alpha=1$ we
did not find a way to employ the communication-complexity method (which tends to result in compact and
``clean'' proofs).
Instead, we design a direct, albeit somewhat involved, proof from first-principles.

Specifically, in order to prove Theorem~\ref{thm:tree-lb}, we consider a complete tree in which each internal vertex has degree $\log n$ (so that its depth is $\Theta\!\left(\frac{\log n}{\loglog n}\right)$). We then consider the family of graphs that
correspond to all possible labelings of such a tree. As noted in Section~\ref{subsubsec:intro-alg}, sampling edges almost uniformly is equivalent to sampling vertices with probability approximately proportional to their degree. In particular, in our construction, the label of the root should be returned with probability approximately $\log n/n$. We show that any algorithm that succeeds in returning the label of the root of the tree with the required probability must perform $\Omega\!\left(\frac{\log n}{\loglog n}\right)$ queries.

To this end we define a {\em process\/} $\calP$ that interacts with any algorithm $\cA$, answering $\cA$'s queries while constructing a uniform random labeling of the vertices and edges in the tree. The vertices of the tree are assigned random labels in $[n]$, and for each vertex $v$ in the tree, its incident edges are assigned random labels in $[d(v)]$.
We say that $\cA$ {\em succeeds\/}, if after the interaction ends, $\cA$ outputs the label of the root of the tree, as assigned by $\calP$.

Let $L = \setL$ be the lower bound we would like to prove, where $C$ is a sufficiently large constant (so that in particular, $L$ is a (small) constant fraction of the depth of the tree).
Intuitively, $\cA$ would like to ``hit'' a vertex at depth  at most $L$ and then ``walk up the tree'' to the root.
There are two sources of uncertainty for $\cA$. One is whether  it actually hits a vertex at depth at most $L$, and the second is which edges should be taken to go up the tree. Our lower bound argument mainly exploits the second uncertainty,
as we sketch next.

The process $P$ starts with an unlabeled tree (of the aforementioned structure), and assigns labels to its vertices and edges in the course of its interaction with $\cA$.
 Recall that $\calP$ answers the queries of $\cA$ while constructing a uniform labeling. Therefore, whenever $\cA$ asks a query involving a \emph{new} label (i.e., that has not yet appeared in its queries or answers to them), the vertex to which this label is assigned, should be uniformly selected  among all vertices that are not yet labeled.
 We shall say that a vertex is \emph{critical} it its depth is at most $L$.
 As long as no critical vertex is hit, $\cA$ cannot reach the root. 
 This implies that if no critical vertex is hit in the course of its queries, then  the probability that
 $\cA$ succeeds is $O(1/n)$.

 While the probability of hitting a critical vertex is relatively small, it is not sufficiently small to be deemed negligible. However, suppose that $\cA$ hits a critical vertex $u$ at depth $\Delta < L$, which occurs with probability $(\log n)^\Delta/n$. Then, conditioned on this event, each of the $(\log n)^\Delta$ edge-labeled paths from $u$ is equally likely to lead to the root, and the labels of vertices on these paths are uniformly distributed, thus intuitively conveying no information regarding the ``right path''.

 A  subtlety that arises when formalizing this argument is the following. Suppose that in addition to hitting  
 a critical vertex $u$, $\cA$ hits another vertex, $v$, which is not necessarily critical, but is at distance less than $L$ from $u$ (and in particular has depth at most $2L$, which we refer to as \emph{shallow}). Then, a path starting from $v$ might meet a path starting from $u$, hence adding a conditioning that makes the above argument (regarding uniform labelings) imprecise. We address this issue by upper bounding the probability of such an event (i.e., of hitting both a critical vertex and a shallow vertex), and accounting for an event of this type as a success of $\cA$.

\subsection{Related work}\label{subsec:intro-related}

Some of the works presented below were already mentioned earlier in the introduction, but are provided in this subsection for the sake of completeness.

The work most closely related to the present work is the recent paper of Eden and Rosenbaum~\cite{Eden2018-sosa}. In~\cite{Eden2018-sosa}, the authors proved
matching upper and lower bounds of $\Theta^*(\ua)$ for the problem of sampling an edge from an almost uniform distribution in an arbitrary graph using degree, neighbor, and pair queries.

The problem of sampling edges in a graph is closely related to the problem of estimating $m$, the number of edges in the graph. In~\cite{Feige2006}, Feige proved an upper bound of $O^*(\ua)$ for obtaining a $(2 + \e)$-factor multiplicative approximation of $m$ using only degree queries,\footnote{To be precise, Feige~\cite{Feige2006} showes that, given a lower bound $d_0$ on the average degree, $O(\sqrt{n/d_0}/\eps)$ degree queries are sufficient. If such a lower bound is not provided to the algorithm, then a geometric search can be performed, as shown in~\cite{Goldreich2008}.} and shows that it is not possible to go below a factor of $2$ with a sublinear number of degree queries. In~\cite{Goldreich2008}, Goldreich and Ron showed that $\Theta^*(\ua)$ queries are necessary and sufficient to obtain a $(1 + \e)$-factor approximation of $m$ if neighbor queries are also allowed.

Several works prove matching upper and lower bounds on the query complexity of counting the number of triangles~\cite{Eden2015}, cliques~\cite{Eden2018-cliques}, and  star-graphs of a given size~\cite{Gonen2011} using degree, neighbor, and pair queries (when the latter are necessary). Eden Ron and Seshadhri devised algorithms for estimating the number of $k$-cliques~\cite{Eden2018-arboricity} and moments of the degree distribution~\cite{Eden2017} 
whose runtimes are parameterized by the arboricity $\alpha$ of the input graph (assuming a suitable upper bound for $\alpha$ is given to the algorithm as input). These algorithms outperform the lower bounds of~\cite{Eden2018-cliques} and~\cite{Gonen2011} (respectively) in the case where $\alpha \ll \sqrt{m}$.
In~\cite{Eden2018-testing-arboricity}, Eden, Levi, and Ron described an efficient algorithm for distinguishing graphs with arboricity at most $\alpha$ from those that are far from any graph with arboricity $3\alpha$.

Two recent works~\cite{Aliakbarpour2017, Assadi2018} consider a query model that allows uniform random edge sampling in addition to degree, neighbor, and pair queries. In this model, Aliakbarpour et al.~\cite{Aliakbarpour2017} described an algorithm for estimating the number of star subgraphs. In the same model, Assadi et al.~\cite{Assadi2018} devised an algorithm that relies on uniform edge samples as a basic query to approximately count the number of instances of an arbitrary subgraph in a graph.
The results in~\cite{Aliakbarpour2017} and~\cite{Assadi2018} imply that uniform edge samples afford the model strictly more power: the sample complexity of the algorithm of~\cite{Aliakbarpour2017} outperforms the lower bound of~\cite{Gonen2011} for the same task, and the sample complexity of the algorithm of~\cite{Assadi2018} outperforms the lower bound of~\cite{Eden2018-cliques} for estimating the number of cliques. (The results of~\cite{Gonen2011, Eden2018-cliques} are in the uniform vertex sampling model.)


\ifnum\icalp=0
\subsection{Organization}
The rest of the paper is organized as follows.
We describe and analyze our main algorithm, thereby proving Theorem~\ref{thm:upper-bound}, in Section~\ref{sec:upper-bound}. The lower bounds of Theorems~\ref{thm:general-lb} and~\ref{thm:tree-lb} are proven in Sections~\ref{sec:general-lower-bound} and~\ref{sec:tree-lower-bound}, respectively.
\else
\subsubsection*{Organization}
Due to space constraints, in this extended abstract we provide full details only for
the proof of our upper bound (Theorem~\ref{thm:upper-bound}).
The proofs of Theorems~\ref{thm:general-lb} and~\ref{thm:tree-lb}
can be found in the appendix.
\fi 


\section{The Algorithm}
\label{sec:upper-bound}
In this section we describe an algorithm that samples an edge $e$ from an arbitrary graph $G$ with arboricity at most $\alpha$, according to a  pointwise almost uniform distribution. Theorem~\ref{thm:upper-bound} follows from our analysis of the algorithm. In what follows, for integers $i \leq j$, we use $[i,j]$ to denote the set of integers $\{i,\dots,j\}$,
and for a vertex $v$ we let $\Gamma(v)$ denote its set of neighbors.

As noted in the introduction,
sampling edges from a uniform distribution is equivalent to sampling vertices proportional to their degrees. Indeed, if each vertex $v$ is sampled with probability $d(v) / 2 m$, then choosing a random neighbor $w \in \Gamma(v)$ uniformly at random returns the (directed) edge $e = (v, w)$ with probability $1/2m$. Thus, it suffices to sample each vertex $v \in V$ with probability (approximately) proportional to its degree.

\subsection{Decomposing graphs of bounded arboricity}
\label{sec:decomposition}

Before describing the algorithm, we describe a decomposition of a graph $G$ into \emph{layers} depending on its arboricity. We begin by recalling the following characterization of arboricity  due to Nash-Williams~\cite{nash1961edge}.

\begin{thm}[Nash-Williams~\cite{nash1961edge}]
  \label{thm:nash-williams}
  Let $G = (V, E)$ be a graph. For a subgraph $H$ of $G$, let $n_H$ and $m_H$ denote the number of vertices and edges, respectively, in $H$. Then
  \confEqn{
  \alpha(G) = \max_{H} \set{\left\lceil m_H / (n_H - 1)\right\rceil},
}
  where the maximum is taken over all subgraphs $H$ of $G$.
\end{thm}

Another folklore result is the connection between the arboricity of a graph and its degeneracy. The degeneracy of  a graph, denoted $\delta$, is the minimum value such that for every subgraph $H\subseteq G$, every vertex $v\in H$ has degree at least $\delta$ (in $H$).

\begin{cor}
  \label{cor:nash-williams}
 For every graph $G$, $\delta\leq 2\alpha(G)$, where $\delta$ is the degeneracy of $G$.
\end{cor}

\begin{dfn}
  Let $G = (V, E)$ be a graph, and $\theta \in \N$, $\beta \in (0, 1)$ parameters. We define a \dft{$(\theta, \beta)$-layering} in $G$ to be the sequence of non-empty disjoint subsets $L_0, L_1, \ldots, L_\ell \subseteq V$ defined by
  \confEqn{
  L_0 = \set{v \in V \sucht d(v) \leq \theta}
}
  and for $i \geq 1$,
  \[
  L_{i+1} = \big\{v \notin L_0 \cup L_1 \cup \cdots \cup L_{i}\, \big |\, \abs{\Gamma(v) \cap \paren{L_0 \cup \cdots \cup L_{i}}} \geq (1 - \beta) d(v)\big\}.
  \]
  That is, $L_0$ consists of all vertices of degree at most $\theta$, and a vertex $v$ is in $L_{i+1}$ if $i$ is the smallest index for which a $(1 - \beta)$-fraction of $v$'s neighbors resides in $L_0 \cup L_1 \cup \cdots \cup L_{i}$. We say that $G$ admits a \dft{$(\theta, \beta)$-layered partition of depth $\mathbf{\ell}$} if we have $V = L_0 \cup L_1 \cup \cdots \cup L_\ell$.
\end{dfn}

\begin{ntn}
  For a fixed $i$, we denote $L_{\leq i} = L_0 \cup L_1 \cup \cdots \cup L_i$, and similarly for $L_{< i}$, $L_{\geq i}$, and $L_{>i}$.  We use the notation $d_i(v)$ to denote $|\Gamma(v) \cap L_i|$ and similarly for $d_{\leq i} (v)$ and $d_{\geq i}(v)$.
\end{ntn}

\begin{lem}
  \label{lem:layered-partition}
  Suppose $G$ is a graph with arboricity at most $\alpha$. Then $G$ admits a $(\theta, \beta)$-layered partition of depth $\ell$ for $\theta = \setthres$, $\beta = \setbeta$, and $\ell \leq \lceil\log n\rceil$.
\end{lem}
\begin{proof}
  For each $i$, let $W_i = V \setminus (L_0 \cup L_1 \cup \cdots \cup L_{i-1})$ be the set of vertices not in levels $0, 1, \ldots, i-1$. Let $m(W_i)$ denote the number of edges in the subgraph of $G$ induced by $W_i$. For any fixed $i$ and $v \in W_{i+1}$, we have $d_{< i}(v) < (1 - \beta) d(v)$ because $v \notin L_{\leq i}$. Therefore, $v$ has at least $\beta d(v) > \beta \thres$ neighbors in $W_i$. Summing over vertices $v \in W_{i+1}$ gives
  \begin{equation}
    m(W_i) = \frac{1}{2}\sum_{v\in W_i}d_{\geq i}(v) \geq \frac{1}{2}\sum_{v\in W_{i+1}}d_{\geq i}(v) > \frac{1}{2} \abs{W_{i+1}}\cdot \beta\thres\;.
    \label{eq:mWi-lb}
  \end{equation}
  On the other hand, since $G$ has arboricity at most $\alpha$, Theorem~\ref{thm:nash-williams} implies that
\ifnum\icalp=0
  \begin{equation}
    m(W_i) \leq \alpha \abs{W_i}\;.
    \label{eq:mWi-ub}
  \end{equation}
  Combining Equations~(\ref{eq:mWi-lb}) and~(\ref{eq:mWi-ub}), we find that
  \[
\frac{\abs{W_{i+1}}}{\abs{W_i}} \leq \frac{2 \alpha}{\beta \theta} = \frac 1 2\;,
\]

  \else
  $ m(W_i) \leq \alpha \abs{W_i}\;.  \label{eq:mWi-ub}$
  Plugging this into Equations~(\ref{eq:mWi-lb}), we find that
  $\frac{\abs{W_{i+1}}}{\abs{W_i}} \leq \frac{2 \alpha}{\beta \theta} = \frac 1 2\;,$
  \fi
  where the inequality is by the choice of $\beta$ and $\theta$. Therefore, $\ell \leq \setell$, as required.

  To see that $V = L_0 \cup L_1 \cup \cdots \cup L_\ell$ (i.e., $W_{\ell+1} = \varnothing$) suppose to the contrary that there exists $v \in W_{\ell+1}$. Then every $v \in W_{\ell+1}$ has at least $\beta d(v) \geq \beta \theta = 4 \alpha$ neighbors in $W_\ell$. By Corollary~\ref{cor:nash-williams}, this implies that $\alpha(G) \geq 2 \alpha$, which is a contradiction.
\end{proof}

\subsection{Algorithm description}

The algorithm exploits the structure of graphs $G$ with arboricity at most $\alpha$ described in Lemma~\ref{lem:layered-partition}. More precisely, as the algorithm does not have direct access to this structure, 
the structure is used explicitly only in the analysis of the algorithm.
Let $L_0, L_1, \ldots, L_\ell$ be a $(\theta, \beta)$-layered partition of $V$ with $\theta = 4 \alpha \lceil\log n \rceil/ \e$, $\beta = \e / \log n$, and $\ell = \log n$. Vertices $v \in L_0$ are sampled with probability exactly proportional to their degree using a simple rejection sampling procedure, $\sal(G, \thres)$. In order to sample vertices in layers $L_i$ for $i > 0$, our algorithm performs a random walk starting from a random vertex in $L_0$ chosen with probability proportional to its degree. Specifically, the algorithm $\se(G, \alpha)$ chooses a length $j$ to the random walk uniformly  in $[0,\ell]$ . The subroutine $\rw(G, \theta, j)$ performs the random walk for $j$ steps, or until a vertex $v \in L_0$ is reached in some step $i > 0$. If the walk returns to $L_0$, the subroutine aborts and does not return any vertex. (This behavior ensures that samples are not too biased towards vertices in lower layers.) Otherwise, $\rw$ returns the vertex at which the random walk halts. Our analysis shows that the probability that the random walk terminates at any vertex $v \in V$ is approximately proportional to $d(v)$ (Corollary~\ref{cor:se}), although $\se$ may fail to return any edge with significant probability. Finally, we repeat $\se$ until it successfully returns a vertex.

\begin{figure}[htb!] \label{se}
  \fbox{
    \begin{minipage}{0.95\textwidth}
      $\se(G, \alpha, \e)$
      \smallskip
      \begin{enumerate}
      \item Let $\theta=\setthres$ and let $\ell=\setell$. \label{step:set}
      \item Choose a number $j \in [0,\ell]$ uniformly at random.
      \item Invoke $\rw(G,\theta,j)$ and let $v$ be the returned vertex if one was returned. Otherwise, \textbf{return} \FAIL. \label{step:invoke_rw}
      \item Sample a uniform neighbor $w$ of $v$ and \textbf{return} $e=(v,w)$.
      \end{enumerate}
    \end{minipage}
  }
\end{figure}

\begin{figure}[htb!] \label{rw}
  \fbox{
    \begin{minipage}{0.95\textwidth}
      $\rw(G, \theta, j)$
      \smallskip
      \begin{enumerate}
      \item Invoke \sal$(\theta)$ and let $v_0$ be the returned vertex if one was returned. Otherwise, \textbf{return} \FAIL.
      \item For $i=1$ to $j$ do
	\begin{enumerate}
	\item Sample a random neighbor $v_i$ of $v_{i-1}$.
	\item If $v_i \in L_0$ then \textbf{return} \FAIL. \label{step:L0_fail}
	\end{enumerate}	
      \item \textbf{Return} $v_j$.	
      \end{enumerate}
    \end{minipage}
  }
\end{figure}

\begin{figure}[htb!] \label{sal}
  \fbox{
    \begin{minipage}{0.95\textwidth}
      $\sal(G, \thres)$
      \smallskip
      \begin{enumerate}
      \item Sample a vertex $u \in V$ uniformly at random and query for its degree. \label{step:l1}
      \item If $d(u) > \thr$ \textbf{return} \FAIL. \label{step:fail_not_L0}
      \item \textbf{Return} $u$ with probability $\frac{d(u)}{\thres}$, and with probability  $1-\frac{d(u)}{\thres}$ \textbf{return} \FAIL.
      \end{enumerate}
    \end{minipage}
  }
\end{figure}


\begin{dfn} \label{def:P}
We let $\HP_j[v]$ denote the probability that \rw\ returns $v$, when invoked with parameters
$G$, $\theta$ and $j\in [0,\ell]$.
We also let $\HP_{\leq j}[v] \eqdef \sum_{i=0}^{j} \HP_i[v]$  and similarly for $\HP_{\geq j}$.	
\end{dfn}

\begin{lem}\label{lem:ub}
  Let $\ell$ be as set in Step~\ref{step:set} of \se\ and let $\HP_{\leq j}$ be as defined in Definition~\ref{def:P}. For all $v\in V$, $\HP_{\leq \ell}[v] \leq \frac{d(v)}{n \thres}$.
\end{lem}
\begin{proof}
  We argue by induction on $j$ that for any  $j \in [0,\ell]$, $\HP_{\leq j}[v] \leq d(v) / n \theta$. For the case $j = 0$, it is immediate from the description of \rw\ and \sal\ that $\HP_0[v] = d(v) / (n \theta)$ if $v \in L_0$ and $\HP_0[v] = 0$ otherwise.
  Further, for $v \in L_0$, due to Step~\ref{step:L0_fail}, $\HP_i(v) = 0$ for all $i > 0$, so that the lemma holds for all $v \in L_0$. Now suppose that for all $v\in V$ we have $\HP_{\leq j-1}(v) \leq d(v) / n \thres$. Then for any fixed $v \notin L_0$ we compute
  \begin{align*}
    \HP_{\leq j}[v] &= \sum_{i = 1}^j \HP_i[v]
    \;=\; \sum_{i = 1}^j \sum_{u \in \Gamma(v)} \HP_{i-1}[u] \frac{1}{d(u)} 
    \;=\; \sum_{u \in \Gamma(v)} \frac{1}{d(u)} \sum_{i = 0}^{j-1} \HP_i[u]\\ 
    &=\; \sum_{u \in \Gamma(v)} \frac{1}{d(u)} \HP_{\leq j-1}[u] 
    \;\leq\; \sum_{u \in \Gamma(v)} \frac{1}{d(u)} \frac{d(u)}{n \thres} 
    \;=\; \frac{d(v)}{n \thres}.
  \end{align*}
  The second equality holds by the definition of \rw, and the one before the last inequality holds by the inductive hypothesis.
\end{proof}



 \begin{lem}\label{lem:lb}
Let $\ell$ be as set in Step~\ref{step:set} of \se.
   For every $j \in [\ell]$, $v\in L_j$ and $k \in [j, \ell]$, we have $\HP_{\leq k}[v] \geq \frac{(1-\beta)^j d(v)}{n\theta}$.
\end{lem}
\begin{proof}
  We prove the claim by induction on $j$. For $j=0$ and $k=0$, by the description of \rw\ and \sal, for every $v\in L_0$,
  \begin{equation}
    \HP_0[v]=\frac{d(v)}{n \thres}\;.\label{eqn:P0}
  \end{equation}
  For $j=0$ and $0<k \leq \ell$,
  \begin{equation}
  \HP_{\leq k}[v]=\sum_{i=0}^{k} \HP_{i}[v] = \HP_0[v]+ \sum_{i=1}^{k}\HP_{i}[v]=\frac{d(v)}{n\theta}\;, \label{eqn:P}
  \end{equation}
  where the last equality is due to Step~\ref{step:L0_fail} in $\rw$.

  For $j=1$ and $1 \leq k \leq \ell$, for every $v \in L_1$, according to Step~\ref{step:L0_fail} in the procedure \rw, $\HP_0[v] = 0$. Also, for every $u \notin L_0$, $\HP_0[u]=0$, since by Step~\ref{step:fail_not_L0} in \sal\ it always holds that $v_0$  is in $L_0$. Therefore,
  \begin{equation}
    \HP_{1}[v]\;=\;\sum_{u \in \Gamma(v) } \HP_0[u] \cdot \frac{1}{d(u)} \;= \;\sum_{u \in \Gamma(v)\cap L_0} \HP_0[u] \cdot \frac{1}{d(u)}\;=\;\sum_{u \in \Gamma(v)\cap L_0} \frac{d(u)}{n\theta} \cdot \frac{1}{d(u)}=\frac{d_0(v)}{n\thres} \label{eqn:pleqj},
  \end{equation}
where the second to last inequality is by Equation~\eqref{eqn:P0}. By the definition of $L_1$, for every $v\in L_1$, $d_0(v)\geq (1-\beta)d(v)$, and it follows that
  \(
  \HP_{\leq k}[v] \geq \HP_{1}[v]\geq (1-\beta)d(v)/ (n\thres)
  \).

  We now assume that the claim holds for all $i \leq j-1$ and $k \in [i, \ell]$, and prove that it holds for $j$ and for every $k \in [j, \ell]$.  By the induction hypothesis and the definition of $L_j$, for any $v \in L_j$ we have
  \begin{align*}
    \HP_{\leq k}[v]&\geq \HP_{\leq j}[v]
    \;=\;\sum_{u \in \Gamma(v)}P_{\leq j-1}[u]\cdot \frac{1}{d(u)}
    \;\geq\; \sum_{i=0}^{j-1}\sum_{u \in \Gamma(v)\cap L_i}\HP_{\leq j-1}[u]\cdot \frac{1}{d(u)}\\
    &\;\geq \; \sum_{i=0}^{j-1}\sum_{u \in \Gamma(v)\cap L_i} \frac{(1-\beta)^{i}d(u)}{n\thres}\cdot \frac{1}{d(u)}
    \geq\;\frac{(1-\beta)^{j-1} d_{\leq j-1}(v)}{n\thres}\geq \frac{(1-\beta)^j d(v)}{n\thres}.
  \end{align*}
  Hence, the claim holds for every $j \in [\ell]$ for every $k \in [j, \ell]$.
\end{proof}

\begin{cor}\label{cor:se}
  For any graph $G$ with arboricity at most $\alpha$,
   the procedure \se\, when invoked with $G$, $\alpha$ and $\e$, returns each edge in the graph with probability in the range $\left[\frac{1-\e/2}{\rho}, \frac{1}{\rho}\right]$ for $\rho=n\thres (\ell +1)$, $\theta = \setthres$ and $\ell=\setell$.
\end{cor}
\begin{proof}
  Consider a specific edge $e^*=(v^*,w^*)$, and let $i$ be the index such that $v^* \in L_i$.
By the description of the procedure \se, the procedure \rw\ is invoked with an index $j$ that is
chosen uniformly in $[0,\ell]$. Hence, the probability that $v^*$ is returned by \rw\ in Step~\ref{step:invoke_rw} is
\[
\Pr[v=v^*]= \frac{1}{\ell+1}\sum_{j=0}^{\ell} P_{j}[v] = \frac{1}{\ell+1} \HP_{\leq \ell}[v].
\]
By Lemma~\ref{lem:ub}, 
$\HP_{\leq \ell}[v] \leq \frac{d(v)}{n \thres }$, and by Lemma~\ref{lem:lb}, 
$\HP_{\leq \ell}[v^*] \geq \frac{(1-\beta)^{\ell}d(v^*)}{n \thres }$, where the probability is over the random coins of the procedures \se\ and \rw. Hence,
\[
\Pr[v=v^*] \in [(1-\beta)^{\ell},1]\cdot \frac{d(v^*)}{n \thres (\ell + 1)},
\]
implying that for $\rho=n\theta(\ell+1)$,
\begin{equation}
\Pr[(v^*, w^*)\text{ is the returned edge}] \in  [(1 - \beta)^{\ell},1] \cdot \frac{1}{n \thres (\ell + 1)}
\in \left[\frac{1-\eps/2}{\rho},\frac{1}{\rho}\right],
\end{equation}
where the last inequality is by the setting of $\beta=\setbeta$.
\end{proof}


\begin{proof}[Proof of Theorem~\ref{thm:upper-bound}]
  Consider the algorithm that repeatedly calls $\se(G, \alpha)$ until an edge $e$ is successfully returned. For a single invocation of $\se$ and fixed edge $e$ let $A_e$ denote the event that $\se$ returns $e$. By Corollary~\ref{cor:se} we have that  $\Pr[A_e] \geq (1 - \e) / n \theta (\ell + 1)$. Further, for any edge $e' \neq e$ the events $A_e$ and $A_{e'}$ are disjoint, so we bound
   \[
    \Pr[\se \text{ returns an edge}] \;=\; \Pr\sqb{\bigcup_{e \in E} A_e}
    \;= \; \sum_{e \in E} \Pr[A_e]
    \;\geq\; \frac{(1 - \e) m}{n \theta (\ell+1)}\;.
\]
  The expected number of iterations until $\se$ succeeds is the reciprocal of this probability, so
  \[
  \E[\# \text{ invocations until success}] \leq \frac{n \theta (\ell+1)}{(1 - \e) m} = O\paren{\frac{n \alpha}{m \e} \cdot \log^2 n}.
  \]
  Since each invocation of $\se$ uses $O(\log n)$ queries, the expected number of queries before an edge is returned is $O(\frac{n\alpha}{\eps m}\cdot \log^3 n)$.

  Finally, when conditioned on a successful invocation of $\se$, Corollary~\ref{cor:se} implies that for any $e, f \in E$ the probabilities $p_e, p_f$ of returning $e$ and $f$, respectively, satisfy
  \[
  1 - \e/2 \leq \frac{p_e}{p_f} \leq \frac{1}{1 - \e/2} \leq 1 +  \e.
  \]
  Therefore, the induced distribution $P$ over edges returned by a successful invocation of $\se$ is pointwise $\e$ close to uniform, which gives the desired result. 
\end{proof}



\ifnum\icalp=0
\section{General Lower Bound}
\label{sec:general-lower-bound}
In this section, we prove Theorem~\ref{thm:general-lb}, thereby showing that the upper bound implied by the algorithm in Section~\ref{sec:upper-bound} is tight for all values of $\alpha$, up to a poly-logarithmic factor in $n$. The lower bound is robust, as it applies to the easier problem of sampling edges from a distribution that is almost uniform with respect to TVD.


In order to prove a general query lower bound for edge sampling (Theorem~\ref{thm:general-lb}), we employ the method of~\cite{LB-CC:ER18} based on communication complexity. The proof is a natural generalization of Theorem~C.3 in the full version~\cite{LB-CC:ER18:full}. The idea of the proof is to construct an embedding of the two-party disjointness function into $\calG_n^\alpha$ in such a way that (1) any algorithm that samples edges from an almost-uniform distribution reveals the value of the disjointness function with sufficiently large probability, and (2) every allowable query can be simulated in the two-party communication setting using little communication. The following definitions and theorem are taken directly from~\cite{LB-CC:ER18}.

\begin{dfn}
  \label{dfn:embedding}
  Let $\pset \subseteq \set{0, 1}^N \times \set{0, 1}^N$. Suppose $f : \pset \to \set{0, 1}$ is an arbitrary (partial) function, and let $g$ be a Boolean function on $\calG_n^\alpha$. Let $\calE : \set{0, 1}^N \times \set{0, 1}^N \to \calG_n^\alpha$. We call the pair $(\calE, g)$ an \dft{embedding} of $f$ if for all $(x, y) \in \pset$ we have $f(x, y) = g(\calE(x, y))$.
\end{dfn}


\begin{dfn}
  \label{dfn:query-cost}
  Let $q : \calG_n^\alpha \to \set{0, 1}^*$ be a query and $(\calE, g)$ an embedding of $f$. We say that $q$ has \dft{communication cost} at most $B$ and write $\cost_{\calE}(q) \leq B$ if there exists a (zero-error) communication protocol $\Pi_q$ such that for all $(x, y) \in P$ we have $\Pi_q(x, y) = q(\calE(x, y))$ and $\abs{\Pi_q(x, y)} \leq B$.
\end{dfn}

\begin{thm}
  \label{thm:general-query-lb}
  Let $\calQ$ be a set of allowable queries, $f : \pset \to \set{0, 1}$, and $(\calE, g)$ an embedding of $f$. Suppose that each query $q \in \calQ$ has communication cost $\cost_{\calE}(q) \leq B$. Suppose $\mA$ is an algorithm that computes $g$ using $T$ queries (in expectation) from $Q$. Then the expected query complexity of $\mA$ is $T = \Omega(R(f) / B)$.
\end{thm}

We prove Theorem~\ref{thm:general-lb} by applying Theorem~\ref{thm:general-query-lb}, and taking $f$ to be the \dft{disjointness function} defined by $\disj(x, y) = 1$ if $\sum_{i = 1}^N x_i y_i = 0$ and $\disj(x, y) = 0$ otherwise,
where $x, y \in \set{0, 1}^N$. The following fundamental result gives a lower bound on the communication complexity of $\disj$.

\begin{thm}[\cite{Kalyanasundaram1992, Razborov1992}]
  \label{thm:disj-lb}
  The randomized communication complexity of the disjointness function is $R(\disj) = \Omega(N)$. This result holds even if $x$ and $y$ are promised to satisfy $\sum_{i = 1}^N x_i y_i \in \set{0, 1}$---that is, Alice's and Bob's inputs are either disjoint or intersect on a single point.
\end{thm}



\begin{proof}[Proof of Theorem~\ref{thm:general-lb}]
  We employ the method of reduction from communication of the disjointness function described 
  above.
  We first describe an embedding $\calE$ of the disjointness function $\disj$. Fix an arbitrary graph $H$  
  on $n'$ vertices with arboricity $\alpha < n'$ and $m'$ edges. Let $K$ be an arbitrary graph with arboricity at most $\alpha$, $2 m'$ edges, and $2 m' / \alpha$ vertices (for example, taking $K$ to be an $\alpha$-regular graph on $2 m' / \alpha$ vertices). We set $n = 2 n'$ and $m = m'$ or $3 m'$ depending on the input of the disjointness function.

  Let $N = n' \alpha / 2 m'$. Given $x, y \in \set{0, 1}^N$, we define $G(x, y) = (V, E)$ as follows. We partition $V = W_0 \cup W_1 \cup \cdots \cup W_N$ where $\abs{W_0} = n'$ and $\abs{W_i} = 2 m' \alpha / n'$ for every $i>0$. The edge set $E$ is determined as follows: $W_0$ is an isomorphic copy of $H$; for $i \geq 1$, $W_i$ is isomorphic to $K$ if $x_i = y_i = 1$, and is a set of isolated vertices otherwise. We define the (partial) function $g : \calG_n^\alpha \to \set{0, 1}$ defined by $g(G) = 1$ if and only if every $W_i$ is a set of isolated vertices. It is straightforward to verify that the pair $(\calE, g)$ where $\calE(x, y) = G(x, y)$ is an embedding of $\disj$ in the sense of Definition~\ref{dfn:embedding}. Further, each degree, neighbor, and pair query to $G(x, y)$ can be simulated by Alice and Bob using $O(1)$ communication. Therefore, by Theorem~\ref{thm:general-query-lb} and the communication lower bound for $\disj$ (Theorem~\ref{thm:disj-lb}), any algorithm that distinguishes $g(G) = 1$ and $g(G) = 0$ with probability bounded away from $1/2$ requires $\Omega(N) = \Omega(n \alpha / m)$ queries in expectation.

  Finally, we show that an algorithm $\cA$ that samples edges from a distribution that is $\e$-close to uniform in TVD can be used to compute $g$ on the image of $\calE$. The key observation is that when $\disj(x, y) = 0$, a $2/3$ fraction of $G(x, y)$'s edges lie in some $W_i$ for $i \geq 1$, while if $\disj(x, y) = 1$, all of $G$'s edges are in $W_0$. Since $\cA$ samples edges from a $\e$-close to uniform distribution for $\e \leq 1/6$, in the case $\disj(x, y) = 0$, $\cA$ must return an edge in some $W_i$ for $i \geq 1$ with probability at least $2/3 - 1/6 = 1/2$.

  Consider the following algorithm: run $\cA$ twice independently to sample edges $e_1$ and $e_2$ from $G = G(x, y)$ using $2 q$ queries in expecatation. If either $e_i$ is not in $W_0$, output $\disj(x, y) = 0$, otherwise output $\disj(x, y) = 1$. The algorithm is always correct on instances where $\disj(x, y) = 1$. On instances where $\disj(x, y) = 0$, it only fails in the case where $e_1$ and $e_2$ both lie in $W_0$. By the computation in the previous paragraph, this event occurs with probability at most $1/2^2 = 1/4$. Therefore, the expected runtime of $\cA$ satisfies $2 q = \Omega(n \alpha / m)$, as desired.
\end{proof}

\section{Lower Bound for Trees}
\label{sec:tree-lower-bound}

In this section we prove Theorem~\ref{thm:tree-lb}, which asserts a query lower bound of $\Omega\paren{\frac{\log n}{\loglog n}}$ for  any algorithm $\cA$ that samples edges in a tree from a pointwise almost uniform distribution.




In order to prove Theorem~\ref{thm:tree-lb} we consider a family of labeled trees $\mT$, where all the trees in the family have the same topology, and differ only in their vertex and edge labeling. The underlying tree $\tT$ is defined as follows: For $k=\log n$, each internal vertex in the tree is incident to $k$ edges, so that the root $r$ has $k$ children, and every other internal vertex has $k-1$ children. Hence, the depth of the tree is $D = \Theta\paren{\frac{\log n}{\loglog n}}$.
In each labeled tree belonging to $\mT$, the vertices are given pairwise distinct labels in $[n]$. Similarly, for each non-leaf vertex $u \in V$ the (ordered) edges $(u, v)$ incident with $u$ are given pairwise distinct labels in $[k]$. We take $\mT$ to be the set of all possible labelings constructed in this way. Since the trees differ only by their labeling, we will sometimes refer to $\mT$ as a family of {\em tree labelings\/}.

Again we note that
sampling edges from a pointwise almost uniform distribution is equivalent to sampling vertices with probability (approximately) proportional to their degree.
We prove Theorem~\ref{thm:tree-lb} by showing that any algorithm that returns the root $r$ of $\tT$ with probability $\Omega\paren{\frac{k}{n}} = \Omega\paren{\frac{\log n}{n}}$ requires $\Omega\paren{\frac{\log n}{\loglog n}}$ queries. More formally, we will prove the following proposition.

\begin{prop}
  \label{prop:root-lb}
There exists a constant $C>1$ such that for any algorithm $\cA$ using at most $\lb = \frac{\log n}{C \cdot \log \log n}$ queries, we have
  \[
  \Pr\left[\cA \text{ returns } r \text{ on input } \lT\; \right] \leq \frac{ \log n}{n \loglog n},
 \]
  where $\lT$ is chosen according to the uniform distribution in $\mT$.
\end{prop}

\begin{rem}
One may wonder if we can increase the lower bound $L$ in Proposition~\ref{prop:root-lb} by setting
 the upper bound on the probability of returning $r$ in the proposition to $\frac{\log n}{C'\loglog n}$ for a constant $C'>1$.
Indeed we can obtain a lower bound of $\lb = \frac{\log n}{C \cdot W(\log n)}$
where $W(\cdot)$ is the Labmert $W$ function (which in particular satisfies $W(x) = \frac{\ln x}{W(x)}$ so that
it is asymptotically slightly smaller than $\log(\cdot)$). However we believe the difference is negligible and hence not worth the slightly less natural construction.
\end{rem}


Let $\cA$ be any algorithm that returns a vertex $v$ in its input graph. We describe a process $\calP$ that interacts with $\cA$ and answers its queries while constructing a uniformly random labeled tree $\lT$ in $\mT$ on the fly. We say that $\cA$ \emph{succeeds}, if $\cA$ outputs the label of the root of $\lT$. The algorithm $\cA$ is allowed $Q \leq \lb$ queries. For the sake of the analysis, if certain ``bad'' events occur, $\calP$ will reveal the entire labeling $\lT$ to $\cA$. Conditioned on a bad event, we simply bound the probability that $\cA$ succeeds by $1$.

\subsection{Details of how \boldmath{$\calP$} answers queries and constructs \boldmath{$\lT \in \mT$}}

\subsubsection{The partially labeled trees $\tT_t$ and $\aTt$}
The process $\calP$ starts with an initially unlabeled tree, denoted $\tT$ (where each internal vertex has degree $k = \log n$ and its depth is $D =\Theta\paren{\frac{\log n}{\loglog n}}$). We denote the set of vertices of $\tT$ by $\tV$ and the set of ordered edges by $\tE$. We also let $\depth(u)$ denote the depth of vertex $u$ in $\tT$.

Following each query of $\cA$, $\calP$ labels some of the vertices/edges in $\tT$.
After it answers the last query of $\cA$, $\calP$ determines all labels of yet unlabeled vertices and edges in $\tT$, thus obtaining the final labeled tree $\lT$.
The labeling decisions made by $P$ ensure that $\lT$ is uniformly distributed in $\mT$
for any algorithm $\cA$.

We denote the partially labeled tree that $\calP$ holds after the first $t$ queries $q_1,\dots,q_t$  of $\cA$, by $\tT_t$. Thus, $\tT_t = (\tV,\tE,\pi_t,\tau_t)$,
where $\pi_t : \tV \to [n]\cup \{\bot\}$ is the partial labeling function of the vertices and $\tau_t : \tE \to [k]\cup \{\bot\}$ is the partial labeling function of the edges.  The symbol $\bot$ stands for `unlabeled'.  The partial labelings $\pi_t$ and $\tau_t$ have the following properties.
\begin{itemize}
\item \emph{Distinctness:} For each pair of distinct vertices $u,v \in \tV$, if $\pi_t(u) \in [n]$ and $\pi_t(v)\in [n]$, then $\pi_t(u) \neq \pi_t(v)$, and an analogous property holds for labels of edges incident to a common vertex.
\item \emph{Consistency:} For each vertex $u \in \tV$, if $\pi_{t-1}(u) \in [n]$, then $\pi_t(u) = \pi_{t-1}(u)$, and an analogous property holds for edge labels.
\end{itemize}
We let $\tV_t = \{u\in \tV:\; \pi_t(u) \in [n]\}$ denote the set of labeled vertices after the first $t$ queries and $\Pi_t = \{\pi_t(u):\; u \in \tV_t\}$ denotes the set of labels of these vertices.  (We set $\tV_0 = \emptyset$ and $\Pi_0 = \emptyset$.)  In order to distinguish between vertices in $\tT_t$ and their labels, we shall use $u,v,w$ for the former, and $x,y,z$ for the latter.  Note that if $x\in \Pi_t$, then $\pi^{-1}_t(x)$ is well defined.


In addition to the partially labeled tree $\tT_t$, the process $\calP$ maintains an auxiliary tree $\aTt = (\aV,\aE,\pi'_t,\tau'_t)$ of depth $\lb=\setL$ in which each internal vertex has degree $k$.
The role of this tree will become clear shortly. For now we just say that in $\aTt$, $\calP$ maintains labeling information
concerning vertices and edges that are ``close to the root of $T\,$'' (without yet determining their exact identity).

Initially $\aTt$ is unlabeled, so that $\pi'_0(u') = \bot$ for every $u' \in \aV$ and $\tau'_0((u',v')) = \bot$ for every $(u',v')\in \aE$.  The labeling functions $\pi'_t$ and $\tau'_t$ maintain the distinctness and consistency properties as defined for $\pi_t$ and $\tau_t$. Furthermore, the labels of vertices in $\aTt$ are distinct from those in $\tT_t$.  We let $\Pi'_t$ denote the set of labels of vertices in $\aTt$ (i.e., $\Pi'_t = \{\pi'_t(u'):\;u' \in \aV\} \setminus \{\bot\}$).  $P$ maintains the invariant that in $\aTt$, the labeled vertices form a connected subtree (rooted at the root $r'$ of $\aTt$).  This is as opposed to $\tT_t$, in which the labeled vertices may correspond to several connected components.

We view $\calP$ as {\em committing\/} to the labels of vertices in $\tV_t$, in the sense that these will be the labels of the corresponding vertices in the final labeled tree $\lT \in \mT$. On the other hand, for labeled vertices in $\aTt$, there is only a ``partial commitment.'' That is, $\calP$ commits to the (partial) labeling of $\aT$ constructed during the interaction with $\cA$, but does not commit to a correspondence between vertices and edges in $\aT$ and those in $\lT$ until after $\cA$ is finished with its queries. Only at the end of the interaction does $\calP$ choose a random embedding of $\aT$ into $\lT$. The embedding induces a partial labeling on $\lT$ as follows: First, the root $r'$ of $T'$ is
mapped to  a vertex $u_*\in V$ at depth at most $\lb$, and $u_*$ is labeled by $\pi'_t(r')$. Then, the labels of the children of $r'$ are assigned to neighbors of $u_*$, and so on, until all the vertices in $T'_t$ are mapped to the vertices of $T$.


\subsubsection{Answering queries}

We refer to a query $q_t$ involving a label $x_t \notin \Pi_{t-1}\cup \Pi'_{t-1}$ as a \dft{new label} query. In order to simplify the presentation, we assume that for any query involving a new label, $\cA$ first performs a degree query on $x_t$, and following this query we have $x_t \in \Pi_t \cup \Pi'_t$.  A lower bound on the number of queries performed by $\cA$ under this assumption translates to the same lower bound up to a factor of three in the standard model.

We also assume that if a neighbor query $(x_t,i_t)$ was answered by $y_t$, where $y_t$ is a label yet unobserved by $\cA$, then $\calP$ provides $\cA$ with the degree of the vertex labeled by $y_t$ as well as the label of the edge from this vertex to the vertex labeled by $x_t$. Similarly, if a pair query $\pair(x_t,y_t)$ is answered positively, then $\calP$ returns the labels of the edges between the vertices labeled by $x_t$ and $y_t$, respectively.  Clearly, any lower bound that holds under these ``augmented answers'' holds under the standard query model.  It follows that for each neighbor query $q_t = \nbr(x_t,i_t)$ we have that $x_t \in \Pi_{t-1}\cup \Pi'_{t-1}$ and similarly for each pair query $q_t = \pair(x_t,y_t)$, we have that $x_t,y_t \in \Pi_{t-1}\cup \Pi'_{t-1}$. On the other hand, for each degree query $q_t = \deg(x_t)$, we may assume that $x_t \notin \Pi_{t-1}\cup \Pi'_{t-1}$ (since $\cA$ is provided with the degrees of all vertices with labels in $\Pi_{t-1} \cup \Pi'_{t-1}$).

Let $\mT_t$ denote the set of all labeled trees in $\mT$ that are consistent with $T_t$ and $T'_t$. That is, $\mT_t$ is the family of all labeled trees in $\mT$ that are consistent with $T$ and with an embedding of $T'$ into $T$ such that the root of $T'$ is mapped to a vertex of $T'$ whose depth is at most $L$.
In order to simplify the analysis we shall elaborately describe how $\calP$  answers each query. However, as will be evident from our description, it holds that $\calP$'s answer to the $t\th$ query $q_t$ can be viewed as choosing a random labeled tree $\lT_{t} \in \mT_{t-1}$, and answering  $q_{t}$ according to $\lT_{t}$. (Observe that $\lT_t$ is only used to answer $q_{t}$ and is later discarded.) Therefore, at the end of the interaction $\calP$ generates a uniform labeled tree $\lT\in \mT$.\footnote{The described process is equivalent to the following.
Consider all possible answers to the query  $q_t$ and weigh each answer $a_t$ according to the fraction of trees in $\mT_{t-1}$ for which $\answer(q_t)=a_t$. That is, for every possible answer $a_t$, $\Pr[a_t]=\frac{|\mT_{t-1}(a_t)|}{|\mT_{t-1}|}$, where $\mT_{t-1}(a_t)$ is the  subset of labeled trees in $\mT_{t-1}$ in which $\answer(q_t)=a_t$.}
\smallskip

\noindent $\calP$ answers the $t\th$ query $q_t$ of $\cA$ as follows:

\ifnum\icalp=0
\paragraph{Degree queries.}
\else
\paragraph{Degree queries}
\fi 
Recall that by assumption, for each degree query $q_t = \deg(x_t)$ we have $x_t \notin \Pi_t \cup \Pi'_t$. When a degree query is made, $\calP$ first decides whether the depth of the vertex to be labeled by $x_t$ is larger than $2\lb$, at most $\lb$, or in between. (Recall that $L=\setL$.) For each vertex $u \in V$, we say that $u$ is \emph{deep} if $\Delta(u) > 2 \lb$, \emph{shallow} if $\Delta(u) \leq 2 \lb$, and \emph{critical} if $\Delta(u) \leq \lb$. We denote the number of deep, shallow, and critical vertices in $\tT$ by $n_d, n_s$, and $n_c$, respectively. For any $t \leq Q$, we denote the number of deep, shallow, and critical vertices that have been labeled by $\calP$ before $\cA$'s $t\th$ query by $\ell_d(t)$, $\ell_s(t)$, and $\ell_c(t)$, respectively. The count $\ell_s(t)$ also includes the number of labels assigned to the auxiliary tree $\aT$, as these labels will eventually be assigned to shallow vertices.

If the $t\th$ query performed by $\cA$ is a degree query, $\calP$ first flips a coin with \emph{bias} $p(t)$ (i.e., the probability that the result is {\sf heads}) defined by
\[
p(t) \eqdef \frac{n_s - \ell_s(t)}{n - \ell_d(t) - \ell_s(t)}\;.
\]
That is, $p(t)$ is the fraction of unlabeled vertices in $\tT_{t-1}$ that are shallow. If the  coin flip turns out {\sf tails} (with probability $1-p(t)$), then we say that
the outcome of the query \emph{is deep}. In this case
$\calP$ selects an unlabeled deep vertex $u$ uniformly at random, sets $\pi_t(u) = x_t$, and returns the degree of $u$.

If the coin flip turns out {\sf heads}, then we say that the outcome of the query is \emph{shallow}. In this case
$\calP$ first checks if
the outcome of any previous degree query was shallow.
If a previous degree query was shallow, $\calP$
completes the labeling of $T$ as follows. It first selects a random embedding of $T'_t$ into $T_t$, then it selects uniformly at random a shallow unlabeled vertex to be the preimage of $x$ and finally it completes the labeling of $T_t$ uniformly at random to a labeled tree $\lT$. After the completion of the labeling, it gives $\cA$ the label of the root of $\lT$, thereby allowing $\cA$ to succeed. In this case we say that $\cA$ \emph{succeeds trivially}.

Otherwise (the outcome of this query is shallow, and there was no previous shallow outcome), $\calP$ flips another coin, this time with bias
\[
p'(t) \eqdef \frac{n_c - \ell_c(t)}{n_s - \ell_s(t)} = \frac{n_c}{n_s - \ell_s(t)}\;,
\]
That is, $p'(t)$ is the fraction of unlabeled shallow vertices that are also critical, where the equality
is justified as follows.
Given that there was no previous degree query whose outcome was shallow, it holds that for every $v\in \tV_t$, $\Delta(u)> 2L-t$, implying that no critical vertex could have been reached through neighbor queries. Therefore, $\ell_c(t)=0$.
If this coin flip results in {\sf tails} (with probability $1 - p'(t)$), then we say that the outcome of the query is \emph{not critical}.
In this case $\calP$ picks a uniformly random unlabeled vertex $u$ that is shallow, but not critical (i.e., with $\lb < \Delta(u) \leq 2 \lb$), assigns $\pi_t(u) = x_t$, and returns the degree of $u$ (which is necessarily $k$). Finally, in the case where the second coin flip is {\sf heads}, i.e., the outcome \emph{is critical}, $\calP$
does the following. It labels the root $r'$ of $\aT$ with $x_t$, i.e., sets $\pi'_t(r')=x_t$, so that now $\Pi'_t = \{x_t\}$. In this subcase $\calP$ returns that the answer to this degree query is $k$.

Observe that the above process is equivalent to choosing a labeled tree $\lT_{t} \in \mT_{t-1}$ uniformly at random, and answering according to $\pi^{-1}(x_{t})$ in $\lT_{t}$. Then $\mT_{t} \subseteq \mT_{t-1}$ is taken to be the sub-family of labelings consistent with this query answer.

\ifnum\icalp=0
\paragraph{Neighbor queries.} 
\else
\paragraph{Neighbor queries}
\fi
First consider the case that $q_t$ is a neighbor query $\nbr(x_t,i_t)$, and recall that by our assumption on $\cA$, $x_t \in \Pi_{t-1}\cup \Pi'_{t-1}$.  If $x_t \in \Pi_t$, then let $u = \pi_{t-1}^{-1}(x_t)$.  We may assume without loss of generality that there is no edge $(u,v)$ incident to $u$ such that $\tau_{t-1}((u,v)) = i$ (or else $\cA$ does not need to perform this query, as its answer is implied by answers to previous queries).  We may also assume that if $u$ is a leaf, so that it has degree $1$, then $i_t = 1$.

Now $\calP$ selects uniformly at random an edge $(u,v) \in \tE$ such that $\tau_{t-1}((u,v)) = \bot$.  If $\pi_{t-1}(v) \in [n]$, then $\calP$ sets $\tau_t((u,v)) = i_t$, and lets $\tau_t((v,u))$ be a uniformly selected $j$ in $[k]\setminus \{\tau_{t-1}((v,w))\}$. Otherwise ($\pi_{t-1}(v) =\bot$), it sets $\pi_t(v)$ to a uniformly selected label in $[n]\setminus (\Pi_{t-1} \cup \Pi'_{t-1})$ and then proceeds as in the case that $\pi_{t-1}(v) \in [n]$.  Finally, $\calP$ returns $\pi_t(v)$ as well as the degrees of $u$ and $v$ and $\tau_t((v,u))$.

If $x_t \in \Pi'_t$, then for $u' = (\pi'_{t-1})^{-1}(x_t)$, $\calP$ selects uniformly at random an edge $(u',v') \in \aE$ such that $\tau_{t-1}((u',v')) = \bot$.  Here it will always be the case that $\pi'_{t-1}(v') = \bot$ (since the labeling of $\aT$ is always done from the root down the tree), so that $\calP$ sets $\pi'_t(v')$ to a uniformly selected label in $[n]\setminus (\Pi_{t-1} \cup \Pi'_{t-1})$. It then sets $\tau'_t((u',v')) = i_t$, and lets $\tau'_t((v',u'))$ be a uniformly selected $j$ in $[k]$.

\ifnum\icalp=0
\paragraph{Pair queries.} 
\else
\paragraph{Pair queries}
\fi
Next consider the case that $q_t$ is a pair query $\pair(x_t,y_t)$. Recall that $x_t,y_t \in \Pi_{t-1}\cup \Pi'_{t-1}$ by assumption. We may assume without loss of generality that neither $y_t$ was an answer to a previous neighbor query $(x_{t'},i_{t'})$ such that $x_{t'} = x_t$, nor that $x_t$ was an answer to a previous neighbor query $(x_{t'},i_{t'})$ such that $x_{t'} = y_t$ (or else $\cA$ does not need to perform this query, as its answer is implied by answers to previous queries).

 If $x_t,y_t \in \Pi_{t-1}$, then let $u = \pi_{t-1}^{-1}(x_t)$ and $v = \pi_{t-1}^{-1}(y_t)$. If $u$ and $v$ are neighbors in $\tT$, then $\calP$ returns a positive answer. It also selects random labels for $(u,v)$ and $(v,u)$ (among those labels not yet used for edges incident to $u$ and $v$, respectively), updates $\tau_t$ accordingly, and returns these labels to $\cA$. If $u$ and $v$ are not neighbors in $\tT$, then $\calP$ returns a negative answer.

If $x_t \in \Pi_{t-1}$ and $y_t \in \Pi'_{t-1}$ (or vice versa), then $\calP$ returns a negative answer, since for any embedding of $T'$ into $T$ it cannot be the case that $(\pi'_{t-1})^{-1}(y_t)$ will be mapped to a neighbor of $ \pi^{-1}_{t-1}(x)$. This is true since for every $x \in \Pi_{t-1}$, $\Delta(\pi^{-1}(x)) > 2L-t_1$ and  every vertex $u'=(\pi^{'})^{-1}(y)$ for $y \in \Pi'_{t-1}$ will be mapped to a vertex in $T$ whose depth is at most $L+t_2$, for $t_1, t_2$ such that $t_1+t_2=t$. Hence, for any embedding of $T'$ into $T$,  the preimages of  $x_t$ and $y_t$ cannot be neighbors.
Finally,  we may assume that $\cA$ does not perform a pair query $(x_t,y_t)$ where $x_t,y_t \in \Pi'_{t-1}$, since the answer to this query is implied by previous queries.

\subsubsection{The final (fully-labeled) tree  $\lT \in \mT$}\label{subsub:lb-complete-label}

Following the last query of $\cA$, the process $\calP$ 
selects a labeled tree $\lT\in \mT_Q$ uniformly at random.
Observe that if $\Pi'_Q=\emptyset$, then each labeled tree in $\mT_Q$ simply corresponds to a possible completion of the labeling of $T_Q$ to a complete labeling of $T$. If $\Pi'_Q\neq \emptyset$, then the trees in $\mT_Q$ can be viewed as trees obtained by  the following two step process. First choose a uniform embedding $T'$ into  $T$ from all possible embeddings that map the root of $T'$ to a critical vertex in $T$ 
Second, complete the labeling to a complete labeling of $T$.

\subsection{Completing the analysis}
\label{subsec:lb-analysis}

In order to complete the analysis, we must argue
that any $\cA$ succeeds when interacting with $\calP$ is small.


To bound the probability that $\cA$ succeeds, we consider three events:
\begin{itemize}
\item $\Eone$: There was no (degree) query whose outcome was critical.
\item $\Etwo$: There was a degree query whose outcome was critical, and no other degree query whose outcome was shallow.
\item $\Ethree$: There was more than one degree query whose outcome was shallow.
\end{itemize}
The three events above are exhaustive. The following three claims bound the probability that $\cA$ succeeds in each case.

\begin{claim}\label{clm:Eone}
  Let $\cA$ be an algorithm that performs $Q \leq \lb$ queries interacting with $\calP$. Then $\Pr[\cA \text{ succeeds } |\, \Eone] = O(1 / n)$, where the probability is taken over both the randomness of $\calP$ and the randomness of $\cA$.
\end{claim}
\begin{proof}
  We consider two subcases. First, suppose $\cA$ returns the label $x$ of a vertex $v$ that was labeled in $\tT_Q$ (i.e., before $\calP$ completed the partial labeling formed during the interaction, so that $x\in \Pi_Q$). Since
  none of $\cA$'s queries had a critical outcome,
  the root $r$ was unlabeled in $\tT_Q$. Therefore, in this case, $\cA$ succeeds with probability $0$.

  Now consider the case that $\cA$ returns a label $x \notin \Pi_Q$. Since $\calP$ chooses the label of $r$, $\pi(r)$, uniformly from $[n] \setminus \Pi_Q$, the probability that $\cA$ succeeds in this case is $1 / (n - \abs{\Pi_Q}) \leq 1 / (n - Q) = O(1/n)$, which gives the desired result. 
\end{proof}


\begin{claim}\label{clm:Etwo}
  Let $\cA$ be any algorithm that performs $Q\leq \lb$ queries interacting with $\calP$. Then
  \[
  \Pr[\cA \text{ succeeds} \wedge \Etwo] =  \Pr[\text{$\cA$ succeeds} \mid \Etwo]\cdot \Pr[\Etwo] =  O\paren{\frac{Q}{n}},
  \]
  where the probability is taken over both the randomness of $\calP$ and the randomness of $\cA$.
\end{claim}
\begin{proof}
We first observe that
\begin{equation}
  \label{eqn:pr-Etwo}
\Pr[\Etwo] \leq \frac{Q\cdot n_c}{n-Q} = O\!\left(\frac{Q \cdot n_c }{ n}\right)\;.
\end{equation}
Indeed, the outcome of the $t\th$ 
query 
is critical with probability
\begin{eqnarray*}
(1-p(t)) \cdot p'(t)
   &= &\frac{n_s - \ell_s(t)}{n - \ell_d(t) - \ell_s(t)} \cdot \frac{n_c}{n_s - \ell_s(t)} 
   = \frac{n_c}{n - \abs{\Pi_t}}\;,
\end{eqnarray*}
so the expression in Equation~\eqref{eqn:pr-Etwo} follows by taking a union bound over $t = 1, 2, \ldots, Q$.

We now turn to bound that probability that $\mA$ succeeds conditioned on the event $\Etwo$.
First note that if $\cA$ outputs a label not in $\Pi'_Q$, then the probability that it succeeds (in particular conditioned on $\Etwo$) is $O(1/n)$, as in the proof of Claim~\ref{clm:Eone}.

Now assume that $\cA$ outputs a label $x \in \Pi'_Q$.
Recall that $\calP$ completes the labeling of $T$ by choosing a labeled tree in $\mT_Q$ (which is the family of labeled trees that result from embedding $T'_Q$ into $T_Q$ uniformly at random  and then completing the labeling of $T$ uniformly at random).
Recall that the embedding is done by mapping the root $r'$ of $T$ into a uniformly selected critical vertex $u_*$ in $T$.
In order to complete the proof of Claim~\ref{clm:Etwo}, we build on the next claim.

\begin{claim}\label{clm:embed}
Consider embedding $T'$ into $T$ by choosing a critical vertex $u_*$ in $T$ uniformly at random and mapping $r'$ to $u_*$. Then any vertex $u'$ in $T'$ is equally likely to be mapped to $r$.
\end{claim}
We defer the proof of the claim and now continue assuming its correctness.
 Given Claim~\ref{clm:embed}, it immediately follows that
\begin{align}
\Pr[\mA \text{ succeeds} \mid \Etwo] = \Pr[(\pi'_Q)^{-1}(x) \text{ will be mapped to } r \mid \Etwo]= \frac{1}{n_c}\;. \label{eqn:pr_suc_Etwo}
\end{align}
Claim~\ref{clm:Etwo} follows by combining Equation~\eqref{eqn:pr_suc_Etwo} with Equation~\eqref{eqn:pr-Etwo}.
\end{proof}

We now prove Claim~\ref{clm:embed}.
\begin{proof}[Proof of Claim~\ref{clm:embed}]
Fix any choice of $u' \in T'$. Let $\Delta' = \Delta(u')$ be the depth of $u'$ in $T'$, and let
$r'=u'_0,u'_1,\dots,u'_{\Delta'}=u'$ denote the sequence of vertices on the path from $r'$ to $u'$ in $T'$.
In order for $u'$ to be mapped to $r$ in the embedding of $T'$ into $T$ (when $r'$ is mapped to a randomly selected critical vertex $u_*$ in $T$), the following events must occur.
\begin{itemize}
\item First, the vertex $u_*$ in $T$ to which $r'$ is mapped must be at depth $\Delta'$ in $T$. This occurs with probability $\frac{k\cdot (k-1)^{\Delta'-1}}{n_c}$.
\item Second, letting
$u_*= v_0,v_1,\dots,v_{\Delta'}=r$ be the sequence of vertices on the path from $u_*$ to $r$ in $T$,
the following  must hold.
For each $j \in [\Delta']$, the child $u'_j$ of $u'_{j-1}$ must be mapped to the parent, $v_j$, of $v_{j-1}$
(recall that $u'_0=r'$ and $v_0 = u_*$). This occurs with probability $k^{-1}\cdot (k-1)^{-(\Delta'-1)}$ as explained next.
In  the random embedding process, each of the $k$ children of $r'=u'_0$ is equally likely to be mapped to the
parent of $u_*=v_0$, and for $j > 1$, conditioned on $u'_1,\dots,u'_{j-1}$ being mapped to $v_1,\dots,v_{j-1}$, respectively, $u'_j$ is mapped to $v_j$ with probability $\frac{1}{k-1}$.
\end{itemize}
The claim follows since the events are independent.
\end{proof}

\begin{figure}[!htb]
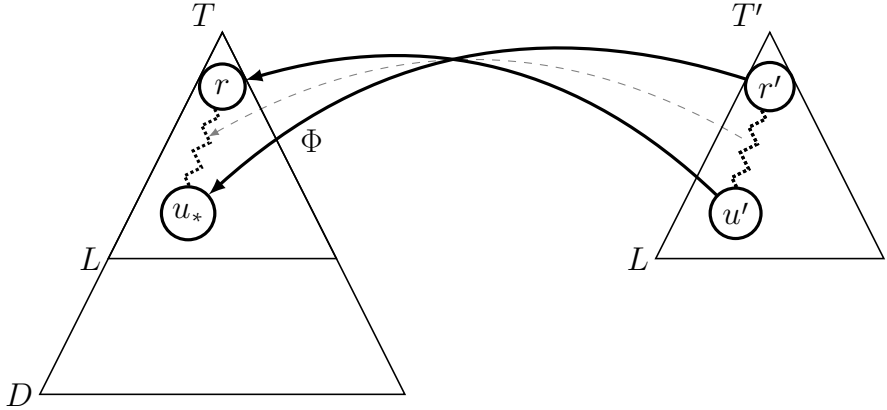

	\centering{
		\DrawTreeLB
	}
	\caption{An illustration of a possible embedding $\phi$ of $T'$ into $T$, where $r'$ is mapped by $\phi$ to $u_*$. Denoting the depth of $u_*$ by $\Delta$, we have  that one of the vertices at depth $\Delta$ in $T'$ will be mapped by $\phi$ to the root $r$ of $T$.}\label{fig:treelb}
\end{figure}

Finally we consider the event $\Ethree$ (conditioned on which $\cA$ trivially succeeds).
\begin{claim}
  \label{clm:Ethree}
 For $Q\leq L$,
\[Pr[\Ethree]=\frac{Q^2k^{4L}}{n^2}\;.\]
\end{claim}
\begin{proof}
  The number of vertices at depth $\Delta \geq 1$ is $k (k - 1)^{\Delta - 1}$, so that the number of shallow vertices is
  \begin{align*}
    n_s &= 1 + k + k (k - 1) + \cdots + k (k - 1)^{2 \lb - 1} = O(k^{2 \lb}).
  \end{align*}
  Thus, the probability that
  the outcome of any degree query
  is shallow is at most $n_s / (n - Q) = O(k^{2 \lb} / n)$.  Taking a union bound over all queries, the probability that at least two are shallow is $O(Q^2 k^{4 \lb} / n^2)$.
\end{proof}

We are now ready to prove Proposition~\ref{prop:root-lb} which implies Theorem~\ref{thm:tree-lb}.
\begin{proof}[Proof of Proposition~\ref{prop:root-lb}]
Since the events $\Eone, \Etwo$ and $\Ethree$ are 
exhaustive,
by Claims~\ref{clm:Eone}, \ref{clm:Etwo} and \ref{clm:Ethree}, it holds for $Q\leq L$ and $L=\setL$ that
\begin{align*}
\Pr[A \text{ succeeds}] &= \Pr[A \text{ succeeds} \mid \Eone]\cdot \Pr[\Eone] + \Pr[A \text{ succeeds} \mid \Etwo]\cdot \Pr[\Etwo]
\\ &+ \Pr[A \text{ succeeds} \mid \Ethree]\cdot \Pr[\Ethree]
\\&= O\paren{\frac{1}{n}}\cdot 1 + O\left(\frac{Q}{n}\right) + 1 \cdot O\paren{\frac{Q^2k^{4L}}{n^{2}}} =O\paren{\frac{Q}{n}}\;.
\end{align*}
Therefore, the proposition follows for $Q=O\paren{\frac{\log n}{\loglog n}}$.
\end{proof}

\fi

%

\bibliography{uniform-arboricity}

\ifnum\icalp=1
\appendix
\section{General Lower Bound}
\label{sec:general-lower-bound}

\section{Lower Bound for Trees}
\label{sec:tree-lower-bound}

\fi

\end{document}